\documentclass[copyright,creativecommons]{eptcs}
\providecommand{\event}{EXPRESS/SOS 2024} 
\DeclareSymbolFont{frenchscript}{OMS}{ztmcm}{m}{n}
\DeclareMathSymbol{\Pow}{\mathord}{frenchscript}{80}  
\usepackage{xcolor}
\usepackage{iftex}
\usepackage{scalerel}
\usepackage{mathpartir}
\ifpdf
  \usepackage{underscore}         
  \usepackage[T1]{fontenc}        
\else
\fi
\usepackage{xspace}
\usepackage{amsthm}
\usepackage{soul}
\usepackage{stmaryrd}
\usepackage{hyperref}
\usepackage{MnSymbol}
\usepackage{amsmath}
\usepackage{amsfonts}
\usepackage{enumitem}

\usepackage{tikz}
{}
{}
{}
\newcommand{\Next}{\operatorname{%
  \tikz[baseline]{
    \draw[line width=.12ex]
      (0,.6ex) circle (.8ex);
  }}}{}
  \newcommand{\Eventually}{\operatorname{%
  \tikz[baseline]{
    \draw[line width=.12ex,line join=round]
      (0ex,.6ex) -- (.95ex,1.55ex) -- (1.9ex,.6ex) -- (.95ex,-.35ex) -- cycle;
  }}}{}
  \newcommand{\Always}{\operatorname{%
  \tikz[baseline]{
    \draw[draw=white] (-0.2ex,0ex) -- (1.7ex,0ex);
    \draw[line width=.12ex,line join=round]
      (0ex,-.2ex) -- (0ex,1.3ex) -- (1.5ex,1.3ex) -- (1.5ex,.-.2ex) -- cycle;
  }}}{}
   \DeclareMathOperator*{\bigdunion}{%
       \scalerel*{\bigcup\mathchoice%
           {\hspace*{-0.75em}\scalebox{0.7}{$\lightning$}\hspace*{0.35em}}%
           {\hspace*{-0.61em}\scalebox{0.55}{\raisebox{0.2em}{$\lightning$}}\hspace*{0.24em}}%
           {ERROR}%
           {ERROR}}{\bigcup}}
  \newcommand{\Until}{\mathbin{\,\mathcal{U}}}{}
  \newcommand{\Release}{\mathbin{\,\mathcal{R}}}{}
  
  \newcommand{\dunion}{\mathbin{%
    \cup\hspace*{-0.4em}\scalebox{0.4}{\raisebox{0.36em}{$\lightning$}}\hspace*{0.35em}}}{}

  \newcommand{\LT}{linear-time\xspace}
  \newcommand{\LTL}{Linear-time Temporal Logic\xspace}
\newcommand{\plat}[1]{\raisebox{0pt}[0pt][0pt]{#1}}     

\title{Semantics for \LTL with Finite Observations}
\author{Rayhana Amjad
\institute{University of Edinburgh\\ Edinburgh, Scotland}
\email{rayhana.amjad@ed.ac.uk}
\and
Rob van Glabbeek
\institute{University of Edinburgh\\ Edinburgh, Scotland}
\email{rvg@cs.stanford.edu}
\and
Liam O'Connor
\institute{Australian National University\\
Canberra, Australia}
\email{liam.oconnor@anu.edu.au}
}
\def\titlerunning{Semantics for \LTL with Finite Observations}
\def\authorrunning{R.Y. Amjad%
, R.J. van Glabbeek%
\ \&\ L. O'Connor}

\newcommand{\ltlthree}[0]{LTL$_3$\xspace}
\newcommand{\prepend}{\triangleright}
\newcommand{\dclose}{\mathop{\lightning}}
\newcommand{\prefixes}{\mathop{\downarrow}}
\newcommand{\extensions}{\mathop{\uparrow}}
\newcommand{\ftraces}{\Sigma^\ast}
\newcommand{\itraces}{\Sigma^\omega}
\newcommand{\traces}{\Sigma^\infty}

\newcommand{\drop}[2]{#1_{\rvert #2}}
\newcommand{\true}{\textup{\texttt{T}}}
\newcommand{\false}{\textup{\texttt{F}}}
\newcommand{\unknown}{\texttt{?}}
\newcommand{\semantics}[1]{\llbracket #1 \rrbracket}

\newcommand{\dbigunion}{\bigdunion}

\newcommand{\SemOp}[1]{%
  \tikz[baseline]{
    \draw[draw=white] (-0.2ex,0ex) -- (1.7ex,0ex);
    \draw[line width=.12ex,line join=round, rounded corners=0.2em]
      (0ex,-.2ex) -- (0ex,1.3ex) -- (1.5ex,1.3ex) -- (1.5ex,.-.2ex) -- cycle;
    \node at (0.75ex,0.55ex) {\tiny $#1$};
  }}

\newcommand{\Nextaif}{\operatorname{\SemOp{\Next}}}{}
\newcommand{\Negaif}{\operatorname{\SemOp{\neg}}}{}
\newcommand{\Trueaif}{\SemOp{\top}}
\newcommand{\Propaif}[1]{\left\ullcorner #1\right\ulrcorner}
\newcommand{\Untilaif}{\mathbin{\SemOp{\Until}}}
\newcommand{\Conjaif}{\mathbin{\SemOp{\land}}}
\newcommand{\Disjaif}{\mathbin{\SemOp{\lor}}}
\newcommand{\Propaifthree}[1]{\Propaif{#1}_3}

\newtheorem{thm}{Theorem}

\begin{document}

\maketitle

\begin{abstract}
\ltlthree is a multi-valued variant of \LTL for runtime verification applications. The semantic descriptions of \ltlthree in previous work are given only in terms of the relationship to conventional LTL\@. Our approach, by contrast, gives a full model-based inductive accounting of the semantics of \ltlthree, in terms of families of \emph{definitive prefix sets}. We show that our definitive prefix sets are isomorphic to \LT temporal properties (sets of infinite traces), and thereby show that our semantics of \ltlthree directly correspond to the semantics of conventional LTL\@. In addition, we formalise the formula progression evaluation technique, popularly used in runtime verification and testing contexts, and show its soundness and completeness up to finite traces with respect to our semantics. All of our definitions and proofs are mechanised in Isabelle/HOL.
\end{abstract}
\section{Introduction}

\LTL (LTL)~\cite{ltl} is one of the most commonly-used logics for the specification of reactive systems. It adds to propositional logic temporal modalities to describe \emph{behaviours}: completed,
infinite traces describing the execution of a system over time. In the context of runtime monitoring or testing, however, we can only make finite observations, and must therefore turn to variants of LTL with finite traces as models. The oldest such variant, commonly attributed to Pnueli\footnote{Such logics are found in many early papers on LTL with Pnueli as a
coauthor such as Lichtenstein et al.~\cite{fltl}, but Manna and Pnueli~\cite{ltlsafety}, which
is usually cited, does not mention finite traces at all.}, concerns finite or infinite \emph{completed} traces, but this is also not suitable for the context of runtime monitoring, as our finite observations are not completed traces, but finite prefixes of infinite behaviours: \emph{partial} traces. 

Bauer et al.~\cite{ltl3tosem}  describe a variant of LTL for partial traces
called {\ltlthree} that distinguishes between those formulae that can be \emph{definitively} said to be true or false from just the partial trace provided, and
those formulae which are \emph{indeterminate}, requiring further states to evaluate definitively. As
we shall see in Section~\ref{sec:ltl}, the semantics of \ltlthree in the literature are
given only in terms of conventional LTL, and Bauer et al.~\cite{bauercomparing} further claim that
\ltlthree \emph{cannot} be given an inductive semantics, a claim that is refuted by the present paper.

We give a compositional, inductive semantics for \ltlthree, in terms of families of \emph{definitive prefix sets}: sets of all (finite or infinite) traces which are sufficient to definitively establish or refute the given formula. We introduce the concept of definitive prefix sets in Section~\ref{sec:dps}, and our semantics in Section~\ref{sec:sems}. We show that definitive prefix sets are determined uniquely by their infinite traces, i.e., that our definitive prefix sets are isomorphic to \LT temporal properties, and thereby we show that the semantics of conventional LTL and of \ltlthree correspond directly. \ltlthree, then, can be understood merely as a different presentation of conventional LTL. 

In Section~\ref{sec:fp} we turn to \emph{formula progression}, a popular technique for evaluating formulae against a finite trace where the formula is evaluated state-by-state, in a style reminiscent of operational semantics or the Brzozowski derivative. Bauer and Falcone~\cite{fpltl3fm} claim without proof that formula progression yields an equivalent semantics to \ltlthree. In this paper, we make this statement formally precise, and prove soundness and completeness (modulo a sufficiently powerful simplifier) of the formula progression technique with respect to our semantics.

Finally in Section~\ref{sec:disc}, we relate our work to other characterisations of prefixes, traces, and properties, as well as to other multi-valued variants of LTL.  All of our work has been mechanised in the Isabelle/HOL proof assistant, proofs of which are available for download~\cite{ourproofs}.

\section[LTL]{\LTL}
\label{sec:ltl}

\begin{figure}
\begin{center}
$$
\begin{array}{llcl}
	\text{Formulae}&\varphi, \psi & ::= & \top \mid a \\
	& & \mid & \neg \varphi \\
	& & \mid & \varphi \land \psi \\
	& & \mid & \Next \varphi \\
	& & \mid & \varphi \Until \psi \\
\text{Atomic propositions} & a & \in & A \\
\text{Traces} & t, u & \in & \traces \\
\text{States} & \Sigma & = & \Pow(A) \\[1em]
\end{array}
$$
Abbreviations:
$$
\begin{array}{lcl}
	\bot & \triangleq & \neg \top \\
	\varphi \lor \psi & \triangleq & \neg (\neg \varphi \land \neg \psi) \\
	\Eventually \varphi & \triangleq & \top \Until \varphi \\
	\varphi \Release \psi & \triangleq & \neg (\neg \varphi \Until \neg \psi)\\
	\Always \varphi & \triangleq & \bot \Release \varphi
\end{array}
$$
\end{center}
\caption{Syntax of LTL}	
\label{fig:syntax}
\end{figure}

Figure~\ref{fig:syntax} describes the syntax of LTL formulae and adjacent definitions. LTL extends propositional logic over \emph{states} (sets of atomic propositions) with temporal operators to produce a logic over \emph{traces}, sequences of states. We denote the set of all states as $\Sigma$. 
A trace $t$ may be finite (in $\ftraces$) or infinite (in $\itraces$). We denote the set of all traces, i.e. $\ftraces \cup \itraces$, as $\traces$. Two traces $t$ and $u$ may be concatenated in the obvious way, written as $tu$. If $t$ is infinite, then $tu = t$. The empty trace is denoted $\varepsilon$.

Our formulation takes conjunction, negation, atomic propositions and the temporal operators \emph{next} ($\Next$) and \emph{until} ($\Until$) as primitive, with disjunction and the temporal operators \emph{eventually} ($\Eventually$), \emph{always} ($\Always$) and \emph{release} ($\Release$) derived from these primitives. 

Figure~\ref{fig:ltlsemantics} gives the semantics of conventional LTL, as a satisfaction relation
whose models are infinite traces. Here $t_0$ denotes the first state of a trace $t$. As we only include future temporal operators, we can advance to the future by dropping initial prefixes from the trace. The notation $\drop{t}{n}$ denotes the trace $t$ without the first $n$ states. If $n$ is greater than the length of $t$, the result of $\drop{t}{n}$ is the empty trace $\varepsilon$. Note that our \emph{until} operator ($\Until$) is \emph{strong}, in that $\varphi \Until \psi$ requires that $\psi$ eventually becomes true at some point in the trace. 
\begin{figure}
\begin{displaymath}
	\boxed{\itraces \models \varphi}
\end{displaymath}
\begin{displaymath}
\begin{array}{lcl}
    t \models \top \\
	t \models a & \text{iff} & a \in t_0 \\
	t \models \neg \varphi & \text{iff} & t \nmodels \varphi\\
	t \models \varphi \land \psi & \text{iff} & t \models \varphi\ \text{and}\ t \models \psi\\
	t \models \Next \varphi & \text{iff} & \drop{t}{1} \models \varphi \\
	t \models \varphi \Until \psi & \text{iff} & \text{there exists}\ i\ \text{s.t.}\ \drop{t}{i} \models \psi\ \text{and}\ \forall j < i.\ \drop{t}{j} \models \varphi
\end{array}	
\end{displaymath}
\caption{Semantics of conventional LTL}	
\label{fig:ltlsemantics}
\end{figure}

Bauer et al.~\cite{ltl3tosem} describe \ltlthree as a \emph{three-valued} logic that interprets LTL formulae on finite prefixes to obtain a truth value in $\mathbb{B}_3 = \{\true,\false,\unknown\}$. For a formula $\varphi$ and a finite prefix $t$, the truth value $\true$ indicates that $\varphi$ can be definitively established from $t$ alone, whereas $\false$ indicates that $\varphi$ can be definitively refuted from $t$ alone. The third value $\unknown$ indicates that the formula $\varphi$ can neither be established nor refuted from $t$ alone:
$$
[t \models_3 \varphi] = \begin{cases} \true & \text{if}\ \forall u \in \itraces.\ tu \models \varphi \\
    \false & \text{if}\ \forall u \in \itraces.\ tu \nmodels \varphi  \\
    \unknown & \text{otherwise}\end{cases}
$$%
Because the truth value $\unknown$ indicates merely that neither $\true$ nor $\false$ apply, \ltlthree can be better understood as a two-valued \emph{partial logic}~\cite{partiallogic}, where $\unknown$ indicates the \emph{absence} of a truth value. In this view, \ltlthree only gives truth values when the trace is \emph{definitive}, i.e., when the answer given will not change regardless of how the trace is extended.

Bauer et al.~\cite{ltl3tosem} note that this presentation of  \ltlthree is inherently non-inductive, i.e., the answer given for a compound formula cannot be produced by combining the answers for its components.  To see why, consider the formula $\varphi = \Eventually a \lor \Eventually \neg a $. Using the semantics above, we have $[\varepsilon \models_3 \varphi ] = \true$ but each component of $\varphi$ produces no definitive answer for the empty trace, i.e. $[\varepsilon \models_3 \Eventually a] = [\varepsilon \models_3 \Eventually \neg a] = \unknown$. Likewise both components of the formula $\psi = \Eventually b \lor \Eventually \neg c$ produce the answer $\unknown$ for the empty trace, but unlike $\varphi$, we have $[ \varepsilon \models_3 \psi] = \unknown$. Therefore, there is no way to combine two $\unknown$ answers in a disjunction that produces correct answers for both $\varphi$ and $\psi$. Because of this, Bauer et al.~\cite{ltl3tosem} claim that inductive semantics are \emph{impossible} for \ltlthree. As we shall see, however, this claim applies only to the multi-valued semantics defined above. In our development, which associates sets of traces to each formula, the semantics for a given formula can indeed be compositionally constructed from the semantics of its components.

\section{Definitive Prefix Sets}
\label{sec:dps}
In this section, we develop a theory of \emph{definitive prefixes}, which we will use in Section~\ref{sec:sems} to give a semantics to \ltlthree. We denote the set of \emph{prefixes} of a trace $t$ as $\prefixes t$:
\[\prefixes t \triangleq \{ u \mid \exists v \in \traces.\ t = uv \}\]
We also generalise this notation to sets, so $\prefixes X$ is the set of all prefixes of traces in $X$. The set of all \emph{extensions} of a trace $t$ is likewise written as $\extensions t$:
\[
\extensions t \triangleq \{ tu \mid u \in \traces \}
\]
The set of \emph{definitive prefixes} of a set of traces $X$ is written $\dclose X$. This is the set of all traces for which all extensions are a prefix of a trace in $X$.
\[
\dclose X \triangleq \{ t \mid \extensions t \subseteq \prefixes X \}
\]
Equivalently, $t$ is a definitive prefix of $X$ iff $X$ contains all infinite extensions of $t$: $\dclose X = \{ t \mid \extensions t \cap \Sigma^\omega \subseteq X \}$.\linebreak[4]
Intuitively, this means that $\dclose X$ contains all those traces for which reaching $X\cap\Sigma^\omega$ is in some way \emph{inevitable}, even if it hasn't happened yet. Definitive prefixes are therefore similar to the notion of \emph{good} and  \emph{bad} prefixes from Kupferman and Vardi~\cite{kupfermanvardi}, just without any moral judgement (see Section~\ref{sec:goodbad}). 

 A set $X$ of traces is called \emph{definitive} iff $X = \dclose X$. Let $\mathcal{D} \subseteq \Pow(\traces)$ denote the set of all definitive sets.

For any set of traces $X$, we have the following straightforwardly from the definitions:
\begin{itemize}
\item All definitive prefixes are prefixes, i.e. $\dclose X \subseteq \prefixes X$.
\item The set $\dclose X$ itself is definitive, i.e. $\dclose \dclose X = \dclose X$.
\item Any extension of a definitive prefix is also a definitive prefix, i.e. $ \forall t \in \dclose X.\ \uparrow t \subseteq \dclose X$.
\item The definitive prefix operator $\lightning$ distributes over intersection, i.e $\lightning (\bigcap_{i \in I} X_i) = \bigcap_{i \in I} \lightning X_i$.
\end{itemize}
The sets $\emptyset$ and $\traces$ are both definitive, and the definitive sets are closed under intersection, i.e., for a set of definitive sets $S$, $\lightning \bigcap S = \bigcap S$. This follows from the distributivity theorem above. The definitive sets are not closed under union, however. To see why, consider when $\Sigma = \{\texttt{A}, \texttt{B}\}$ and the set $X_\texttt{A}$ contains all traces starting with $\texttt{A}$ and the set $X_\texttt{B}$ contains all traces starting with $\texttt{B}$. The sets $X_\texttt{A}$ and $X_\texttt{B}$ are both definitive, but their union is not: neither $X_\texttt{A}$ nor $X_\texttt{B}$ contain the empty trace $\varepsilon$, but $ \varepsilon \in \dclose (X_\texttt{A} \cup X_\texttt{B})$, as all extensions of $\varepsilon$ (i.e. all non-empty traces) must begin with either $\texttt{A}$ or $\texttt{B}$.

\subsection{Lattice Properties}
Define the \emph{definitive union}, written $\dbigunion S$ or $X \dunion Y$ in the binary case, as merely the definitive prefixes of the union: 
$$
\dbigunion S \triangleq \dclose \bigcup S
$$
\begin{thm}The definitive union gives least upper bounds for definitive sets ordered by set inclusion, i.e., for a set $S \subseteq \mathcal{D}$ of definitive sets:
\begin{itemize}
\item For all $X \in S$, $X \subseteq \bigdunion S$.
\item If there is a definitive set $Z$ such that $\forall X \in S.\ X \subseteq Z$, then $\bigdunion S \subseteq Z$.
\end{itemize}
\end{thm}
\begin{proof}Follows from definitions.	
\end{proof}

\noindent Thus, the definitive sets $\mathcal{D}$ ordered by set inclusion form a complete lattice, where the supremum is the definitive union, the infimum is the intersection, the greatest element $\top$ is $\traces$ and the least element $\bot$ is $\emptyset$.

\subsection[Isomorphism to LTL Properties]{Isomorphism to Linear-time Temporal Properties}

\newcommand{\Prop}{\mathsf{Pr}}
\newcommand{\Definitives}{\mathsf{Df}}
\begin{thm}\label{thm:bijection}Define the lower adjoint $\Prop : \mathcal{D} \rightarrow \Pow(\itraces)$ as $\Prop(X) = X \cap \itraces$, and the upper adjoint $\Definitives : \Pow(\itraces) \rightarrow \mathcal{D}$ as $\Definitives(P) = \dclose P$. We have, for any definitive set $X$ and \LT temporal property $P$: $$ \Prop(X) = P\ \text{if and only if}\ X = \Definitives(P)$$\end{thm}
\begin{proof} Proving each direction separately:
\begin{description}
	\item[$\implies$] It suffices to show that $\Definitives(\Prop(X)) = X$, i.e.\ $ \lightning(X \cap \itraces) = X$. Because any extension of a definitive prefix is also a definitive prefix, and all traces in $X$ are definitive prefixes, every finite trace in $X$ must be a prefix of some infinite trace in $X$. Thus, the infinite traces in $X$ alone are sufficient to describe all their definitive prefixes, i.e., all traces in $X$.
	\item[$\impliedby$] It suffices to show that $\Prop(\Definitives(P)) = P$, i.e.\ $\lightning P \cap \itraces = P$. Recall that the definitive prefixes of $P$ are all those traces $t$ for which all extensions of $t$ are a prefix of a trace in $P$. For an infinite trace $t \in P$, 
	the set of extensions $\extensions t$ is just $\{ t \}$, which is surely contained in $\prefixes P$. Therefore, for a \LT temporal property $P$ (consisting only of infinite traces), we can conclude $P \subseteq \lightning P$. In fact, $\lightning P$ consists of all the (infinite) traces of $P$ as well as possibly some finite prefixes of these. Hence $\lightning P \cap \itraces = P$.
\popQED
\end{description}	
\end{proof}
\begin{thm}
	$\Prop$ (and likewise for $\Definitives$) is monotone and preserves least upper and greatest lower bounds, i.e:
	\begin{itemize}
		\item If $A \subseteq B$ then $\Prop(A) \subseteq \Prop(B)$
		\item $\Prop(\bigcap_{i \in I} X_i) = \bigcap_{i \in I} \Prop(X_i)$
		\item $\Prop(\dbigunion_{i \in I} X_i) = \bigcup_{i \in I} \Prop(X_i)$
	\end{itemize}
\end{thm}
\begin{proof}
The first two statements follow directly from definitions. Preservation of least upper bounds requires more finesse. As we have already seen in the proof of Theorem~\ref{thm:bijection}, for any set of traces $S$, $\Prop(\lightning S) = \Prop(S)$. This means that $\Prop(\dbigunion_{i \in I} X_i) = \Prop(\bigcup_{i \in I} X_i) = \bigcup_{i \in I} \Prop(X_i)$ as required.
\end{proof}
\noindent These theorems say that $(\Prop, \Definitives)$ forms a \emph{lattice isomorphism} between definitive sets and \LT temporal properties.

\section[Semantics of LTL and LTL3]{Semantics of LTL and \ltlthree}
\label{sec:sems}

\subsection{Answer-indexed Families}

We give a semantics to \ltlthree by compositionally assigning to each formula
an \emph{answer-indexed family} of definitive sets. In general, an \emph{answer-indexed family} is a function that, given an answer (e.g.~a value in $\mathbb{B} = \{ \true, \false \}$), produces a set of models (depending on the logic, this could be a set of states, a definitive set, a linear-time temporal property, etc.). This set contains all those models which produce the given answer for the formula in question. In this way, we invert the traditional presentation of multi-valued logics, where the truth value is the output of a satisfaction function, and instead take the desired answer $a$ as an \emph{input} and produce models as an \emph{output}: $a = [\sigma \models \varphi]$ becomes $\sigma \in \semantics{\varphi}\ a$. An answer-indexed family for conventional, infinite-trace LTL therefore produces a linear-time temporal property as an output:
\[\Phi,\Psi \; \in \;  \mathbb{B} \rightarrow \Pow (\itraces)\]
whereas an answer-indexed family for \ltlthree produces a definitive set as an output:
\[\Phi,\Psi\; \in\; \mathbb{B} \rightarrow \mathcal{D}\]
To begin with, we define an alternative semantics for \emph{conventional} LTL in terms of answer-indexed families of linear-time temporal properties.  This requires us to define various operations on answer-indexed families, one for each kind of LTL constructor:
\begin{align*}
\Trueaif \ \true &= \itraces &(\Phi \Conjaif \Psi) \ \true &= \Phi \ \true \cap \Psi \ \true & (\Phi \Disjaif \Psi) \ \true &= \Phi \ \true \cup \Psi \ \true & (\Negaif \Phi) \ \true &= \Phi \ \false \\
\Trueaif \ \false &= \emptyset &(\Phi \Conjaif \Psi) \ \false &= \Phi \ \false \cup \Psi \ \false & (\Phi \Disjaif \Psi) \ \false &= \Phi \ \false \cap \Psi \ \false & (\Negaif \Phi) \ \false &= \Phi \ \true 
\end{align*}
Note that the set operations used for the $\false$ answer are always the duals of the operations used for the $\true$ answer, which means that for conventional LTL, the set produced for the $\false$ answer is always the complement of the set for the $\true$ answer. This means that the operator for negation ($\Negaif$) can simply swap the places of the set and its complement. This is akin to performing a conversion to negation normal form ``just-in-time'' as we evaluate a formula.

\noindent For atomic propositions $a$, the corresponding answer-indexed family maps $\true$ to the set of all traces that begin with a state containing $a$, and $\false$ to its complement:
\begin{align*}
\Propaif{\emph{a}} \ \true &= \{t \mid t \in \itraces \land \emph{a} \in t_0 \} \\
\Propaif{\emph{a}} \ \false &= \{t \mid t \in \itraces \land  \emph{a} \notin t_0\}
\end{align*}
The semantic operator for $\Next \varphi$ formulae prepends one state to all the corresponding traces for $\varphi$, analogously to the conventional LTL semantics in Figure~\ref{fig:ltlsemantics}:
\begin{align*}
(\Nextaif \Phi) \ \true &= \{t \mid  \drop{t}{1} \in \Phi \ \true \} \\
(\Nextaif \Phi) \ \false &= \{t \mid  \drop{t}{1} \in \Phi \ \false \}
\end{align*}
The $\true$ case of the semantic operator for $\varphi \Until \psi$ formulae is also defined analogously to Figure~\ref{fig:ltlsemantics}, with the $\false$ case being the complement:
\begin{align*}
(\Phi \Untilaif \Psi) \ \true &= \{t \mid \exists k. (\forall i < k. \drop{t}{i} \in \Phi \ \true) \wedge \drop{t}{k} \in \Psi \ \true \} \\
(\Phi \Untilaif \Psi) \ \false &= \{t \mid \forall k. (\exists i < k. \drop{t}{i} \in \Phi \ \false) \vee \drop{t}{k} \in \Psi \ \false \}
\end{align*}
Finally, we put all of these semantic operators to use in Figure~\ref{aifsemantics}, which gives a compositional, inductive semantics to conventional LTL using these operators.

\begin{figure}
\centering
\begin{align*}
\semantics{\top} \ &=\ \Trueaif \\
\semantics{\emph{a}} \ &=\ \Propaif{\emph{a}} \\
\semantics{\neg \varphi} \ &=\ \Negaif \semantics{\varphi}\\
\semantics{\varphi \wedge \psi} \ &=\ \semantics{\varphi} \Conjaif \semantics{\psi}\\
\semantics{\varphi \vee \psi} \ &=\ \semantics{\varphi} \Disjaif \semantics{\psi} \\
\semantics{\Next \varphi} \ &=\ \Nextaif \semantics{\varphi} \\
\semantics{\varphi \Until \psi} \ &=\ \semantics{\varphi} \Untilaif \semantics{\psi}
\end{align*}

\caption{LTL semantics using answer-indexed families}
\label{aifsemantics}
\end{figure}

\begin{thm}[Equivalence to conventional semantics]\label{thm:ltlequiv}
Answer-indexed family LTL semantics assigns the same truth values to a given trace for a given formula as conventional LTL semantics:
\begin{itemize}
\item $(t \vDash \varphi) \Longleftrightarrow (t \in \semantics{\varphi}\ \true)$
\item $\neg (t \vDash \varphi) \Longleftrightarrow (t \in \semantics{\varphi}\ \false)$
\end{itemize}
\end{thm}
\begin{proof}
This is proven straightforwardly by induction on $\varphi$, justified in the same way as a conversion to negation normal form.
\end{proof}

\subsection{The Prepend Operation}
To give a semantics to \ltlthree, our answer-indexed families will produce definitive sets, rather than linear-time temporal properties. To this end, we will define an auxiliary operation on definitive sets called \emph{prepend}, written $\prepend X$, which gives all traces whose tails are in $X$:
$$\prepend X\; \triangleq \; \{ t \mid \drop{t}{1} \in X \}$$

\begin{thm}\label{thm:prependclosed}
The \emph{prepend} operation is closed for definitive sets. That is, if $X$ is definitive, then $\prepend X$ is definitive.
\end{thm}
\begin{proof}
We must show that $\lightning (\prepend X) = \prepend X$ for any definitive set $X$. Showing each direction separately:
\begin{description}
	\item[$\implies$] Given a definitive prefix $t \in \lightning(\prepend X)$, we must show that $t \in \prepend X$. If $t = \varepsilon$, then this implies that $\prepend X = \traces$ and therefore $t \in \prepend X$. If $t = \sigma{}u$, because $\prefixes (\prepend X) = \prepend (\prefixes X)$, we can conclude \linebreak $\extensions \sigma{}u \subseteq \prepend (\downarrow X) $. Taking the tail of both sides, we can see that $\extensions u \subseteq {\prefixes X}$ and therefore $u \in X$ as $X$ is definitive. Prepending $\sigma$ to both sides, we conclude that $\sigma{}u \in \prepend X$ as required.
	\item[$\impliedby$] Given a prefix $t \in \prepend X$, we must show that $t \in \lightning (\prepend X)$. If $t = \epsilon$, this means that $X = \prepend X = \lightning (\prepend X) = \traces$ as $X$ is definitive. If $t = \sigma{}u$, we know that $u \in X$. As $X$ is definitive, all extensions of $u$ are also in $X$. Therefore $\prepend (\extensions u) \subseteq X$ and thus $\sigma{}u \in \lightning (\prepend X)$.
\popQED
\end{description}
\end{proof}

\subsection[Semantics for LTL3]{Semantics for \ltlthree}
The semantic operators for \ltlthree resemble that of conventional LTL, except that now we work with definitive sets rather than linear-time temporal properties. 
\begin{align*}
\Trueaif _3 \ \true &= \traces &(\Phi \Conjaif_3 \Psi) \ \true &= \Phi \ \true \cap \Psi \ \true & (\Phi \Disjaif_3 \Psi) \ \true &= \Phi \ \true \dunion \Psi \ \true & (\Negaif _3 \Phi) \ \true &= \Phi \ \false \\
\Trueaif _3 \ \false &= \emptyset &(\Phi \Conjaif_3 \Psi) \ \false &= \Phi \ \false \dunion \Psi \ \false & (\Phi \Disjaif_3 \Psi) \ \false &= \Phi \ \false \cap \Psi \ \false & (\Negaif _3 \Phi) \ \false &= \Phi \ \true
\end{align*}
All of the sets produced by these answer-indexed families are definitive, as $\traces$ and $\emptyset$ are both definitive sets and definitive sets are closed under intersection and definitive union. Unlike with conventional LTL, the set for the $\false$ answer is not the complement of the set for the $\true$ answer, as definitive sets are not closed under complement. The set for $\true$ contains all traces that are sufficient to definitively satisfy the formula, and the set for $\false$ contains all traces that are sufficient to definitively refute the formula.

For an atomic proposition $a$, the set for $\true$ contains all non-empty traces that begin with a state that satisfies $a$, and the set for $\false$ contains all non-empty traces that begin with a state that does not satisfy $a$. However, if $a$ is trivial, in the sense that \emph{all} or \emph{no} possible states satisfy $a$, then these sets are not definitive, as the excluded empty trace $\varepsilon$ would also be definitive for these sets. Thus, we take the definitive prefixes of these sets to account for this possibility:%
\begin{align*}
\Propaifthree{\emph{a}} \ \true &= \lightning\{t \mid t \neq \epsilon \wedge  \emph{a} \in t_0 \} \\
\Propaifthree{\emph{a}} \ \false &= \lightning\{t \mid t \neq \epsilon \wedge  \emph{a} \notin t_0\} \end{align*}
For the $\Next$ operator, we make use of the prepend operator, which by Theorem~\ref{thm:prependclosed} produces definitive sets:
\begin{align*}
(\Nextaif _3 \Phi) \ \true &= \prepend \ (\Phi \ \true) \\
(\Nextaif _3 \Phi) \ \false &= \prepend \ (\Phi \ \false)
\end{align*}
For the $\Until$ operator, we construct our semantics iteratively, building up by repeatedly prepending states. Here the notation $f^k$ indicates the self-composition of $f$ $k$ times, i.e.~$f^0(x) = x$ and $f^{k+1}(x) = f^k(f(x))$:
\begin{align*}
(\Phi \Untilaif _3 \Psi) \ \true &= \bigdunion_{k \in \mathbb{N}} f^k(\Psi \ \true), \ \text{where} \ f(X) = \prepend X \cap \Phi \ \true\\
(\Phi \Untilaif _3 \Psi) \ \false &= \bigcap_{k \in \mathbb{N}} f^k(\Psi \ \false), \ \text{where} \ f(X) = \prepend X \dunion \Phi \ \false
\end{align*}
Because definitive sets are closed under intersection, definitive union and the prepend operator, we can see that that our $\Untilaif$ operator also produces definitive sets by a simple inductive argument on the natural number $k$.  Using all of these operations, we construct an inductive, compositional semantics for \ltlthree  in Figure~\ref{ltl3semantics}.  
\begin{figure}
\centering
\begin{align*}
\semantics{\top} _3 \ &=\ \Trueaif _3 \\
\semantics{\emph{a}} _3 \ &=\ \Propaifthree{\emph{a}} \\
\semantics{\neg \varphi} _3 \ &=\ \Negaif_3\ \semantics{\varphi} _3\\
\semantics{\varphi \wedge \psi} _3 \ &=\ \semantics{\varphi} _3  \Conjaif _3  \semantics{\psi} _3\\
\semantics{\varphi \vee \psi} _3 \ &=\ \semantics{\varphi} _3  \Disjaif _3  \semantics{\psi} _3 \\
\semantics{\Next \varphi} _3 \ &=\ \Nextaif_3\ \semantics{\varphi} _3 \\
\semantics{\varphi \Until \psi} _3 \ &=\ \semantics{\varphi}_3 \Untilaif_3 \semantics{\psi}_3 
\end{align*}

\caption{\ltlthree semantics using answer-indexed families}
\label{ltl3semantics}
\end{figure}
\begin{thm}[Equivalence to original \ltlthree definition]\label{thm:ltl3old}
	Let $t$ be a finite prefix and $\varphi$ be an LTL formula. Then:
\begin{itemize}
	\item $t \in \semantics{\varphi}_3\ \true \iff \forall u \in \itraces.\ tu \in \semantics{\varphi}\ \true$
	\item $t \in \semantics{\varphi}_3\ \false \iff \forall u \in \itraces.\ tu \in \semantics{\varphi}\ \false$
\end{itemize}
\end{thm}
\begin{proof}
	Follows directly from the definition of definitive sets, as $\semantics{\varphi}_3\ \true$ and $\semantics{\varphi}_3\ \false$ are both definitive.
\end{proof}
\noindent Theorem~\ref{thm:ltl3old} shows that our inductive semantics coincides with the original non-inductive semantics given for \ltlthree. If we view our semantics through the lens of the isomorphism in Theorem~\ref{thm:bijection}, however, we see that this semantics is also equivalent to the semantics of conventional LTL:
\begin{thm}[Equivalence to conventional LTL]\label{thm:ltl3equiv}
For all formulae $\varphi$:
\begin{itemize}
\item $\Prop (\semantics{\varphi} _3\ \true) = \semantics{\varphi}\ \true$
\item $\Prop (\semantics{\varphi} _3\ \false) = \semantics{\varphi}\ \false$
\end{itemize}
\end{thm}
\begin{proof}
The two statements are shown simultaneously by induction on $\varphi$:
\begin{itemize}[leftmargin=7em]
	\item[$\varphi = \top$:]	$\Prop(\traces) = \itraces$ by definition.
	\item[$\varphi = a$:]  Because $\Prop(\lightning S) = \Prop(S)$ as seen in the proof of Theorem~\ref{thm:bijection}, $\Prop(\Propaifthree{a} \true) = \Propaif{a}\ \true$ and likewise $\Prop(\Propaifthree{a} \false) = \Propaif{a}\ \false$. 
	\item[$\varphi = \neg \varphi'$:] Follows from inductive hypotheses. 
	\item[$\varphi = \varphi' \land \psi'$:] Follows from inductive hypotheses as $\Prop$ preserves greatest lower  and least upper bounds.
	\item[$\varphi = \Next \varphi'$:] Follows from inductive hypotheses as the prepend operator $\prepend$ commutes with $\Prop$.
	\item[$\varphi = \varphi' \Until \psi'$:] Because $\Prop$ commutes with $\prepend$ and preserves least upper and greatest lower bounds,\linebreak we can show that $\Prop(\bigdunion_{k \in \mathbb{N}} f^k(\semantics{\psi'}_3 \ \true))$, where $f(X) = \prepend X \cap \semantics{\varphi}_3 \ \true$, is equal to\linebreak $\bigcup_{k \in \mathbb{N}} g^k(\Prop(\semantics{\psi'}_3 \ \true))$ where  $g(X) = \prepend X \cap \Prop(\semantics{\varphi'}_3 \ \true)$ by induction on on the natural number $k$. By the inductive hypotheses, this is equal to $\bigcup_{k \in \mathbb{N}} g^k(\semantics{\psi'} \ \true)$ where  $g(X) = \prepend X \cap \semantics{\varphi'} \ \true$.   This can be shown by another simple induction to be equal to the original definition in the conventional LTL semantics $\{t \mid \exists k. (\forall i < k. \drop{t}{i} \in \semantics{\varphi'} \ \true) \wedge \drop{t}{k} \in \semantics{\psi'} \ \true \}$. The cases for the $\false$ answer are proved similarly.
\popQED
\end{itemize}
\end{proof}
\noindent Because of this equivalence theorem, we can now express the relationship between the set for $\true$ and the set for $\false$ in our \ltlthree semantics. In \ltlthree, while the two sets do not overlap, they are not perfect complements of each other as they were in conventional LTL, as definitive sets are not closed under complement. Instead, the $\false$ set is the definitive set corresponding to the \emph{linear-time temporal property} containing all infinite traces not in the $\true$ set.
\begin{thm}[Excluded Middle]\label{thm:em}
For all formulae $\varphi$: $$ \semantics{\varphi}_3\ \true = \lightning(\itraces \setminus \semantics{\varphi}_3\ \false) \qquad \qquad \text{and} \qquad \qquad \semantics{\varphi}_3\ \false = \lightning(\itraces \setminus \semantics{\varphi}_3\ \true) $$	
\end{thm}
\begin{proof} It is a straightforward consequence of Theorem~\ref{thm:ltlequiv} that $\semantics{\varphi}\ \true = \itraces \setminus \semantics{\varphi}\ \false$ (*). Then:
\begin{displaymath}
	\begin{array}{lclr}
		 \semantics{\varphi}_3\ \true & =& \Definitives(\Prop ( \semantics{\varphi}_3\ \true)) & \text{(Theorem~\ref{thm:bijection})}\\
         & =& \Definitives(\semantics{\varphi}\ \true) & \text{(Theorem~\ref{thm:ltl3equiv})}\\
         & =& \Definitives(\itraces \setminus \semantics{\varphi}\ \false) & \text{(*)}\\	 
         & =& \Definitives(\itraces \setminus \Prop (\semantics{\varphi}_3\ \false)) & \text{(Theorem~\ref{thm:ltl3equiv})}\\
         & =& \Definitives(\itraces \setminus (\semantics{\varphi}_3\ \false \cap \itraces)) &\qquad \text{(Definition of $\Prop$)}\\         
         & =& \lightning (\itraces \setminus \semantics{\varphi}_3\ \false) & \text{ (Definition of $\Definitives$)}\\[-4ex]
	\end{array}
\end{displaymath}	
\end{proof}

\section{Formula Progression}
\label{sec:fp}
\begin{figure}
\begin{gather*}
\boxed{\varphi \xrightarrow{\sigma} \psi}\\
\inferrule{ }{\top \xrightarrow{\sigma} \top} \qquad 
\inferrule{a \in \sigma}{a \xrightarrow{\sigma} \top} \qquad 
\inferrule{a \notin \sigma}{a \xrightarrow{\sigma} \bot} \qquad 
 \inferrule{ }{\Next \varphi \xrightarrow{\sigma} \varphi} \\
 \inferrule{ \varphi \xrightarrow{\sigma} \varphi' }{\neg \varphi \xrightarrow{\sigma} \neg \varphi'}\qquad 
\inferrule{ \varphi \xrightarrow{\sigma} \varphi' \\
              \psi \xrightarrow{\sigma} \psi' }{\varphi \land \psi \xrightarrow{\sigma} \varphi' \land \psi'} \qquad 
\inferrule{ \psi \xrightarrow{\sigma} \psi' \\ \varphi \xrightarrow{\sigma} \varphi'}{\varphi \Until \psi \xrightarrow{\sigma} \psi' \lor (\varphi' \land (\varphi \Until \psi)) }
\end{gather*}
\caption{Rules for formula progression}
\label{fig:fp}	
\end{figure}

\emph{Formula progression} is a technique first introduced by Kabanza et al.~\cite{progress1,progress2} that evaluates a formula stepwise against states in a style reminiscent of operational semantics or the Brzozowski derivative. This technique was used by O'Connor and Wickstr\"om~\cite{quickstrom} as the basis for their testing algorithm, and by Bauer and Falcone~\cite{fpltl3fm} for decentralised monitoring of component-based systems.  Figure~\ref{fig:fp} gives an overview of formula progression rules for LTL. The judgement \plat{$\varphi \xrightarrow{\sigma} \psi$} states that, to prove $\varphi$, it suffices to prove $\psi$ for the tail of our trace if the head of our trace is $\sigma \in \Sigma$. Note that these rules are total and syntax-directed on the left-hand formula $\varphi$. This means that these rules taken together constitute a definition of a total function that takes $\varphi$ and $\sigma$ as input and produces $\psi$ as output. We generalise this notation to finite prefixes, so that for a finite trace $t = \sigma_0\dots\sigma_n$, the notation $\varphi \xrightarrow{t} \psi$ just means $\varphi \xrightarrow{\sigma_0} \cdots \xrightarrow{\sigma_n} \psi$. Repeated application of these rules, however, can lead to exponential blowup in the size of the formula. While both O'Connor and Wickstr\"om~\cite{quickstrom} and Bauer and Falcone~\cite{fpltl3fm} report that interleaving this progression with formula simplification at each step keeps the formulae tractable for most practical use cases, Ro\c{s}u and Havelund \cite{rosu} warn that pathological exponential cases still exist.

Bauer and Falcone~\cite{fpltl3fm} state that formula progression can serve as an \emph{alternative semantics} for \ltlthree on finite traces, where a formula $\varphi$ is considered definitively true for a finite trace $t$ iff $\varphi \xrightarrow{t} \top$, definitively false iff $\varphi \xrightarrow{t} \bot$, and is unknown otherwise. While it goes unmentioned in their paper, here the implicit simplification steps are not just a performance optimisation, but are vital to ensure that the semantics given via formula progression is complete with respect to the standard \ltlthree semantics. To see why, consider the formula $\Eventually a$. Let $\sigma_a$ be a state where $a \in \sigma_a$. Then the formula $\Eventually a$ should be considered definitively true for the trace consisting of just $\sigma_a$. The formula generated by our formula progression rules, however, would be $\top \lor (\top \land \Eventually a)$, which yields the desired formula $\top$ only after logical simplifications are applied. While in this case, the simplifications required are just identities of propositional logic, in general such straightforward simplifications alone are insufficient. For example, consider the formula $(\Next a) \lor (\Eventually \neg a)$. According to the semantics of \ltlthree presented above, this formula should be considered definitively true for the empty trace $\epsilon$, as it is a tautology. Temporally local simplifications such as those used by O'Connor and Wickstr\"om~\cite{quickstrom}, however, would not be able to determine that this formula is a tautology until after one state has been observed. Therefore, in order for formula progression to align correctly with the semantics of \ltlthree, the simplification must transform \emph{all} tautologies into $\top$ and \emph{all} absurdities into $\bot$. A simple, although slow way to implement such a simplifier would be to convert both the formula and its negation into B\"uchi automata, and perform cycle detection to check for emptiness. For our development, we abstract away from such syntactic simplification procedures by working only on the level of our model-based semantics. As can be seen in our Theorem~\ref{thm:fp3} given below, we do not seek a specific syntactic tautology $\top$ or absurdity $\bot$, but rather refer to any formula with trivial semantics. A purely syntactic characterisation, by contrast, would require a full accounting of the simplification procedure, which is outside the scope of our development here.

The rules given in Figure~\ref{fig:fp} operate on one state at at a time, whereas our semantics are on the level of entire traces. Therefore, in order to show soundness and completeness (for finite traces) of our formula progression rules with respect to our semantics, we must first prove two lemmas which relate a single step of formula progression to our semantics.

The first lemma states that for one step of formula progression $\varphi \xrightarrow{\sigma} \varphi'$, prepending $\sigma$ to the traces that satisfy/refute the \emph{output} formula $\varphi'$ yields traces that satisfy (resp.\ refute) the \emph{input} formula $\varphi$.
  
 \begin{thm}\label{thm:fp1} Let $\varphi$ and $\varphi'$ be formulae and $\sigma$ be a state such that $\varphi \xrightarrow{\sigma} \varphi'$. Then:
 \begin{itemize}\item
 	$ \prepend (\semantics{\varphi'}_3\ \true) \cap \{ t \mid t_0 = \sigma \} \subseteq \semantics{\varphi}_3\ \true  $
 	\item 
 	 	$ \prepend (\semantics{\varphi'}_3\ \false) \cap \{ t \mid t_0 = \sigma \} \subseteq \semantics{\varphi}_3\ \false  $
 \end{itemize}
 \end{thm}
\begin{proof}
	The two statements are shown simultaneously by structural induction on the formula $\varphi$ (which, as our rules are syntax directed, uniquely determines the output formula $\varphi'$). The base cases for $\varphi = \top$ and $\varphi = a$ as well as the inductive cases for the next operator $\Next$ follow directly from definitions. Of the other inductive cases, the cases for conjunction and disjunction require the use of the distributive properties of the lattice of definitive sets, as well as the fact that the prepend operator $\prepend$ distributes over intersection and definitive union. The cases for negation follow directly from the inductive hypotheses, whereas the cases for the until operator $\Until$ require unfolding of the big unions and intersections in the definition of the semantic operator $\Untilaif\!_3$ by one step.
\end{proof}
\noindent The second lemma states that those traces that satisfy the \emph{input} formula $\varphi$ and begin with the state $\sigma$ will have tails that satisfy the \emph{output} formula $\varphi'$.
\begin{thm}\label{thm:fp2}  Let $\varphi$ and $\varphi'$ be formulae and $\sigma$ be a state such that $\varphi \xrightarrow{\sigma} \varphi'$. Then:
 \begin{itemize}\item
 	$ \semantics{\varphi}_3\ \true\ \cap \{ t \mid t_0 = \sigma \} \subseteq \prepend (\semantics{\varphi'}_3\ \true)  $
 	\item 
 	 	$ \semantics{\varphi}_3\ \false\ \cap \{ t \mid t_0 = \sigma \} \subseteq \prepend (\semantics{\varphi'}_3\ \false)  $
 \end{itemize}
\end{thm}
\begin{proof}
As with Theorem~\ref{thm:fp1}, the two statements are shown simultaneously by structural induction on the formula $\varphi$. The base cases and the cases for the next operator $\Next$ are shown just by unfolding definitions, the cases for conjunction and disjunction are shown by use of distributive properties including those of the prepend operator $\prepend$, negation proceeds directly from the induction hypotheses, and the until operator $\Until$ requires unfolding of the semantic operator $\Untilaif\!_3$ by one step.
\end{proof}
\noindent By combining these two lemmas, we can inductively prove a theorem that relates formula progression to our semantics on the level of entire finite traces. This resembles the informal definition of formula progression semantics given by Bauer and Falcone~\cite{fpltl3fm}, but with the syntactic requirement that the ultimate formula be $\top$ or $\bot$ replaced by a semantic requirement that it has trivial semantics.
\begin{thm}\label{thm:fp3} Let $t \in \ftraces$ be a finite trace. Then, for all formulae $\varphi$ and $\varphi'$ where $\varphi \xrightarrow{t} \varphi'$:
\begin{itemize}
\item  $t \in \semantics{\varphi}_3\ \true$ if and only if $\semantics{\varphi'}_3\ \true = \traces$.
\item  $t \in \semantics{\varphi}_3\ \false$ if and only if $\semantics{\varphi'}_3\ \false = \traces$.	
\end{itemize}	
\end{thm}
\begin{proof}
By induction on the length of the trace $t$ (where $\varphi$ and $\varphi'$ are kept arbitrary). The second statement for $\false$ is proved identically to the first for $\true$, so we present the proof only for $\true$ here.
	\begin{description}
	\item[Base Case $(t = \varepsilon)$] It suffices to show that $\varepsilon \in \semantics{\varphi}_3\ \true$ iff $\semantics{\varphi}_3\ \true = \traces$. Because $\semantics{\varphi}_3\ \true$ is a definitive set, and any extension of a definitive prefix is also a definitive prefix, as $\varepsilon$ is in $\semantics{\varphi}_3\ \true$, we can conclude that all traces (i.e.\ extensions of $\varepsilon$) are in $\semantics{\varphi}_3\ \true$. The reverse direction of the iff is straightforward.
	\item[Inductive Case $(t = \sigma u)$] We know that $\varphi_0 \xrightarrow{\sigma} \varphi \xrightarrow{u} \varphi'$ and have the inductive hypothesis that $u \in \semantics{\varphi}_3\ \true\linebreak[2] \iff \semantics{\varphi'}_3\ \true = \traces$. We must show that $\sigma u \in \semantics{\varphi_0}_3\ \true \iff \semantics{\varphi'}_3\ \true = \traces$. Therefore, by the inductive hypothesis, it suffices to show $\sigma u \in \semantics{\varphi_0}_3\ \true \iff u \in \semantics{\varphi}_3\ \true$. Showing each direction separately:
	\begin{description}
	\item[$\implies$] By Theorem~\ref{thm:fp2} we can conclude that $\sigma u \in \prepend (\semantics{\phi}_3\ \true)$ and thus that $u \in \semantics{\varphi}_3\ \true$ by the definition of the prepend operator $\prepend$. 
	\item[$\impliedby$] By the definition of the prepend operator $\prepend$ we can conclude that $\sigma u \in \prepend(\semantics{\varphi}_3\ \true)$ and thus that $\sigma u \in \semantics{\varphi_0}_3\ \true$ by Theorem~\ref{thm:fp1}. 
\popQED
	\end{description}
\end{description}	
\end{proof}
\noindent Theorem~\ref{thm:fp3} is both a \emph{soundness} and \emph{completeness} proof for formula progression semantics with respect to our model-based semantics, up to finite traces. Soundness here means that a formula $\varphi$ will only evaluate in formula progression to a tautology for a trace $t$ when $t$ is in $\semantics{\varphi}_3\ \true$, and likewise will only evaluate to an absurdity when $t$ is in $t$ is in $\semantics{\varphi}_3\ \false$. This is the $\impliedby\!\!$ direction of the iff in Theorem~\ref{thm:fp3}. Completeness (or adequacy) up to finite traces means that all finite prefixes that definitively confirm the formula will evaluate in formula progression to a tautology, and all finite prefixes that definitively refute the formula will evaluate to an absurdity. This is the $\implies\!\!$ direction of the iff in Theorem~\ref{thm:fp3}.
\section{Discussion}
\label{sec:disc}

\subsection{Prefix Characterisations and Monitorability}
\label{sec:goodbad}

Kupferman and Vardi~\cite{kupfermanvardi} define the \emph{bad prefixes} of a property $P \subseteq \itraces$ as those finite prefixes that cannot be extended to a trace that is in $P$, and further define \emph{good prefixes} as those for whom all infinite extensions are in $P$. For any \LT temporal property $P$, we can see from our definitions that $\lightning P$ consists of $P$ along with all good prefixes of $P$. The bad prefixes of $P$ can be obtained by taking the definitive prefixes of the complement of $P$, i.e.\ $\lightning (\itraces \setminus P)$. Our answer-indexed families $\mathbb{B} \rightarrow \mathcal{D}$ can be thought of as tracking both the \emph{good} and \emph{bad} prefixes of a property simultaneously, along with the infinite traces that they approximate. That is, for a formula $\varphi$, $\semantics{\varphi}_3\ \true$ contains all infinite traces that satisfy $\varphi$ as well as the good prefixes of $\varphi$, and $\semantics{\varphi}_3\ \false$ contains all infinite traces that do not satisfy $\varphi$ as well as the bad prefixes of $\varphi$.

Bauer et al.~\cite{bauercomparing} further define \emph{ugly prefixes} as those that cannot be finitely extended into good nor bad prefixes.
 Note that the good, bad, and ugly prefixes do not constitute a complete classification of all finite prefixes. For example, the prefix $\textit{ppp}\dots$ is not in $\semantics{ p \Until q }_3\ \true$ nor $\semantics{ p \Until q }_3\ \false$, but it it not ugly either, as it can be extended with $q$ giving a good prefix, or with $\emptyset$ giving a bad prefix. Here, $\emptyset$ is the state satisfying neither $p$ nor $q$. The presence of ugly prefixes means that the formula is \emph{non-monitorable}.
 
Aceto et al.~\cite{aceto2019adventures} define monitorability positively, through a framework for synthesising monitors from modal $\mu$-calculus formulae. The semantics of these monitors resembles our formula progression semantics, and thus it may be interesting to find some connection (such as bisimilarity) between these. 
Aceto et al.~define monitorable formulae as those for which a monitor can be synthesised --- we conjecture that this definition and that of Bauer et al.~\cite{bauercomparing} coincide. They also define syntactic fragments of modal $\mu$-calculus that are monitorable for acceptance and violation, which is a useful syntactic accounting of monitorability that may be transferable to \ltlthree. Like us, Aceto et al.~give their semantics of modal $\mu$-calculus in terms of sets of traces, including sets of both finite and infinite traces (`finfinite' traces) --- they do not, however, consider definitive prefixes, and as such their finfinite semantics does not align with our \ltlthree semantics.
 
\subsection{Safety and Liveness}
Linear-time temporal properties can be broadly categorised into \emph{safety} properties, which state that something ``bad'' does not happen during execution, and \emph{liveness} properties, which state that something ``good'' will eventually happen during execution~\cite{Lam77}. Alpern and Schneider~\cite{alpernschneider} provide a formal characterisation by equipping $\itraces$ with a metric space structure, where the distance between two traces is measured inversely to the length of their longest common prefix.  Then, safety properties are those sets that are limit-closed (i.e.\ $\overline{P} = P$) and liveness properties are those sets that are dense (i.e.\ $\overline{P} = \itraces$). The key insight that enables this elegant characterisation is that a safety property can always be definitively refuted by a finite prefix of a trace, whereas any finite prefix can be extended in such a way as to satisfy a given liveness property. We also see in later work~\cite{kupfermanvardi,quant} the concept of \emph{co-safety} (or \emph{guarantee}) properties and \emph{co-liveness}  (or \emph{morbidity}) properties, the complements of safety and liveness properties respectively. A \emph{co-safety} property can always be definitively \emph{confirmed} by a finite prefix of a trace,\footnote{Thus, \emph{co-safety} could be seen as an alternative formalisation of Lamport's informal concept of a liveness property~\cite{Lam77}, different from the standard formalisation of \cite{alpernschneider}.} whereas any finite prefix can be extended in such a way as to \emph{refute} a given \emph{co-liveness} property.

Our definitive sets include those finite prefixes that can confirm (or refute) the property, enabling us to express these insights about finite prefixes directly. This provides an alternative characterisation that we conjecture is equivalent to that of Alpern and Schneider~\cite{alpernschneider}. 
\begin{description}
\item[Liveness Properties] 	Liveness properties are those that can never be definitively refuted by a finite prefix. Thus a definitive set $X$ represents a liveness property iff all finite traces are prefixes of traces in $X$, i.e. $\ftraces \subseteq \prefixes X$. A co-liveness property can never be definitively confirmed by a finite prefix. As we saw in Theorem~\ref{thm:em}, the complement of a definitive set $X$ is given by $\lightning (\itraces \setminus X)$. This gives us a characterisation of co-liveness, where $X$ represents a co-liveness property iff $\ftraces \subseteq \prefixes (\itraces \setminus X)$.
\item[Safety Properties] Safety properties are those that can always be definitively refuted by a finite prefix, but because our definitive sets include definitive confirmations and not refutations, it is easier to begin with co-safety properties, which can always be definitively \emph{confirmed} by a finite prefix.  That is, any infinite trace in the property must be an extension of some \emph{finite} definitive prefix of the property. Thus, a definitive set $X$ represents a co-safety property iff $ X = \extensions(X \cap \ftraces)$. A safety property is just the complement of a co-safety property, i.e.\ $X$ is a safety property iff $ \lightning (\itraces \setminus X) = \extensions(\lightning (\itraces \setminus X) \cap \ftraces) $.
\end{description}

\subsection{RV-LTL}

\ltlthree only gives \emph{definitive} non-$\unknown$ answers, that is, a formula is judged to be true (resp.\ false) for a finite trace $t$ only if all extensions of that prefix $t$ are also true (resp.\ false). As noted by Bauer et al.~\cite{bauercomparing}, this means that there exists a large class of formulae for which no definitive answers can be given for any finite trace. For example, take the standard \emph{request/acknowledge} format:
$$
\Always (r \Rightarrow \Eventually a)
$$%
which states that all requests ($r$) must eventually be acknowledged ($a$). For every finite prefix $u$, we have $ur^\omega \in \semantics{\varphi}_3\ \false $ and $ua^\omega \in \semantics{\varphi}_3\ \true$. As the $\false$ and $\true$ answers are non-overlapping (Theorem~\ref{thm:em}), $u$ must not be a definitive prefix. Therefore, \emph{all} finite prefixes are not definitive, meaning that \ltlthree cannot give a non-$\unknown$ answer for any finite trace. 
To remedy this, Bauer et al.~\cite{bauergbu, bauercomparing} propose RV-LTL, a dialect of LTL specifically for the domain of runtime verification. RV-LTL is more accurately an ad-hoc layering of \ltlthree on top of Pnueli's LTL for finite traces (here notated $\models_\textsf{F}$). Where \ltlthree would give the $\unknown$ answer, RV-LTL instead gives a \emph{presumptive} answer ($\top^\textsf{p}$ or $\bot^\textsf{p}$ ) based on the answer obtained from Pnueli's finite LTL:
\[
[u \models \varphi]_\textsf{RV} = \begin{cases}
 	\top & \text{if}\ [ u \models \varphi]_3 = \top \\
 	\bot & \text{if}\ [ u \models \varphi]_3 = \bot \\
 	\top^\textsf{p} & \text{if}\ [ u \models \varphi]_3 = \texttt{?}\ \text{and}\ u \models_\textsf{F} \varphi \\
 	\bot^\textsf{p} & \text{if}\ [ u \models \varphi]_3 = \texttt{?}\ \text{and}\ u \nmodels_\textsf{F} \varphi \\
 \end{cases}
\]
Intuitively, after a finite prefix $u$, a definitive answer ($\top$ or $\bot$) is unchangeable no matter how the prefix is extended, whereas a presumptive answer  ($\top^\textsf{p}$ or $\bot^\textsf{p}$ ) only applies if execution is stopped at that point. 

If the property in question is a safety property, then the only presumptive answer possible is $\top^\textsf{p}$, and likewise for co-safety properties and $\bot^\textsf{p}$. This means that for properties at the bottom of the safety-progress hierarchy~\cite{safetyprogress}, \ltlthree is sufficient, as the single $\unknown$ answer can be interpreted as $\top^\textsf{p}$ or $\bot^\textsf{p}$ respectively. However, as noted by Bauer et al.~\cite{bauercomparing}, there are monitorable properties such as $((p \lor q) \Until r) \lor \Always p$ for which both $\top^\textsf{p}$ and $\bot^\textsf{p}$ answers are possible (consider $\emph{qqqq}\dots$ and $\emph{pppp}\dots$).

Like \ltlthree previously, the semantics of RV-LTL is presented only in terms of other logics. We believe that an inductive semantics can be designed along similar principles to that of \ltlthree given in the present paper, where our answer indexed-families instead produce four sets, two of which are definitive, rather than the two definitive sets we provide for \ltlthree.

As noted by O'Connor and Wickstr\"om~\cite{quickstrom}, Pnueli's finite LTL is a logic of finite \emph{completed} traces, so the decision to judge partial traces as completed for the purpose of giving presumptive answers in RV-LTL is ad-hoc and can produce rather arbitrary answers for properties higher in the safety-progress hierarchy.  
For example, consider a system where a flashing light consistently alternates between \textsf{On} and \textsf{Off} states:
$$ \textsf{On}\ \textsf{Off}\ \textsf{On}\ \textsf{Off}\ \cdots $$
A simple property that we might wish to monitor for this system is that the light is $\textsf{On}$ infinitely often:
$$ \Always \Eventually \textsf{On} $$%
As this formula nests $\Always$ and $\Eventually$ operators, it is definitive in neither positive nor negative cases and will only give presumptive answers.  But the
presumptive answer given in RV-LTL depends only on the very last observed status of the light. For a trace where the light continuously alternates off and on, as above, we might intuitively say that presumptive answer ought to be true, but this formula would be considered presumptively false if our observation happens to end in a state where the light is off. Thus, the truth value obtained for this formula is overly sensitive to the point at which our finite observation ceases. 

One potential approach that may provide a more robust logic for finite traces would be to first decompose the property into \ltlthree-monitorable and non-monitorable components, and, where possible, combine the answers obtained by monitoring each monitorable component separately. Such decompositions are very general: for example, Alpern and Schneider~\cite{alpernschneider} famously prove that \emph{all properties} are the intersection of a safety (i.e.\ \ltlthree-monitorable) property and a liveness property. We conjecture that there will be some configuration of this approach whose answers coincide with RV-LTL, but it will be interesting future work to explore the design space here.
\section{Conclusion}
We have presented a new, inductive, model-based semantic accounting of \ltlthree in terms of answer-indexed families of definitive sets, and in the process shown that \ltlthree is more accurately described as a more detailed presentation of conventional LTL, rather than a distinct logic in its own right. We have formalised the popular formula progression technique used in runtime verification and testing scenarios, and proved it sound and complete with respect to our semantics. All of our work has been mechanised in over 1700 lines of Isabelle/HOL proof script.

We anticipate that our theory of definitive sets will provide a semantic foundation for other logics of partial traces, such as the LTL$^\pm$ of Eisner et al.~\cite{ltlpm}, QuickLTL from O'Connor and Wickstr\"om~\cite{quickstrom}, or the aforementioned RV-LTL~\cite{bauercomparing}. Our answer-indexed families may also be applicable to other multi-valued logics. Examples include rLTL~\cite{rltl}, RV-LTL~\cite{bauercomparing}, and the five-valued logic of Chai et al.~\cite{chai}. We intend, in future work, to develop logics that go beyond just the definitive prefixes of \ltlthree, giving presumptive or probabilistic answers when definitive answers are unavailable.
\nocite{*}
\bibliographystyle{eptcsalpha}
\bibliography{cites}

\newcommand{\etalchar}[1]{$^{#1}$}
\begin{thebibliography}{AAF{\etalchar{+}}19}
\providecommand{\bibitemdeclare}[2]{}
\providecommand{\surnamestart}{}
\providecommand{\surnameend}{}
\providecommand{\urlprefix}{Available at }
\providecommand{\url}[1]{\texttt{#1}}
\providecommand{\href}[2]{\texttt{#2}}
\providecommand{\urlalt}[2]{\href{#1}{#2}}
\providecommand{\doi}[1]{doi:\urlalt{https://doi.org/#1}{#1}}
\providecommand{\eprint}[1]{arXiv:\urlalt{https://arxiv.org/abs/#1}{#1}}
\providecommand{\bibinfo}[2]{#2}

\bibitemdeclare{article}{aceto2019adventures}
\bibitem[AAF{\etalchar{+}}19]{aceto2019adventures}
\bibinfo{author}{Luca \surnamestart Aceto\surnameend}, \bibinfo{author}{Antonis
  \surnamestart Achilleos\surnameend}, \bibinfo{author}{Adrian \surnamestart
  Francalanza\surnameend}, \bibinfo{author}{Anna \surnamestart
  Ing{\'o}lfsd{\'o}ttir\surnameend} \& \bibinfo{author}{Karoliina \surnamestart
  Lehtinen\surnameend} (\bibinfo{year}{2019}): \emph{\bibinfo{title}{Adventures
  in monitorability: from branching to linear time and back again}}.
\newblock {\slshape \bibinfo{journal}{Proceedings of the ACM on Programming
  Languages}} \bibinfo{volume}{3}(\bibinfo{number}{POPL}), pp.
  \bibinfo{pages}{1--29}, \doi{10.1145/3290365}.

\bibitemdeclare{article}{ourproofs}
\bibitem[AGO24]{ourproofs}
\bibinfo{author}{Rayhana \surnamestart Amjad\surnameend}, \bibinfo{author}{Rob
  \surnamestart van Glabbeek\surnameend} \& \bibinfo{author}{Liam \surnamestart
  O'Connor\surnameend} (\bibinfo{year}{2024}): \emph{\bibinfo{title}{Definitive
  Set Semantics for LTL3}}.
\newblock {\slshape \bibinfo{journal}{Archive of Formal Proofs}}.
\newblock \bibinfo{note}{\url{https://isa-afp.org/entries/LTL3_Semantics.html},
  Formal proof development}.

\bibitemdeclare{article}{alpernschneider}
\bibitem[AS85]{alpernschneider}
\bibinfo{author}{Bowen \surnamestart Alpern\surnameend} \&
  \bibinfo{author}{Fred~B. \surnamestart Schneider\surnameend}
  (\bibinfo{year}{1985}): \emph{\bibinfo{title}{Defining liveness}}.
\newblock {\slshape \bibinfo{journal}{Information Processing Letters}}
  \bibinfo{volume}{21}(\bibinfo{number}{4}), pp. \bibinfo{pages}{181--185},
  \doi{10.1016/0020-0190(85)90056-0}.

\bibitemdeclare{inproceedings}{fpltl3fm}
\bibitem[BF12]{fpltl3fm}
\bibinfo{author}{Andreas \surnamestart Bauer\surnameend} \&
  \bibinfo{author}{Yli{\`e}s \surnamestart Falcone\surnameend}
  (\bibinfo{year}{2012}): \emph{\bibinfo{title}{Decentralised LTL Monitoring}}.
\newblock In: {\slshape \bibinfo{booktitle}{FM 2012: Formal Methods}},
  \bibinfo{publisher}{Springer}, pp. \bibinfo{pages}{85--100},
  \doi{10.1007/978-3-642-32759-9_10}.

\bibitemdeclare{inbook}{progress1}
\bibitem[BK96]{progress1}
\bibinfo{author}{Fahiem \surnamestart Bacchus\surnameend} \&
  \bibinfo{author}{Froduald \surnamestart Kabanza\surnameend}
  (\bibinfo{year}{1996}): \emph{\bibinfo{title}{Using Temporal Logic to Control
  Search in a Forward Chaining Planner}}, p. \bibinfo{pages}{141–153}.
\newblock \bibinfo{publisher}{IOS Press}.

\bibitemdeclare{inproceedings}{partiallogic}
\bibitem[Bla02]{partiallogic}
\bibinfo{author}{Stephen \surnamestart Blamey\surnameend}
  (\bibinfo{year}{2002}): \emph{\bibinfo{title}{Partial Logic}}.
\newblock In: {\slshape \bibinfo{booktitle}{Handbook of Philosophical Logic}},
  \bibinfo{publisher}{Springer}, pp. \bibinfo{pages}{261--353},
  \doi{10.1007/978-94-017-0458-8\_5}.

\bibitemdeclare{inproceedings}{bauergbu}
\bibitem[BLS07]{bauergbu}
\bibinfo{author}{Andreas \surnamestart Bauer\surnameend},
  \bibinfo{author}{Martin \surnamestart Leucker\surnameend} \&
  \bibinfo{author}{Christian \surnamestart Schallhart\surnameend}
  (\bibinfo{year}{2007}): \emph{\bibinfo{title}{The Good, the Bad, and the
  Ugly, But How Ugly Is Ugly?}}
\newblock In: {\slshape \bibinfo{booktitle}{Runtime Verification}},
  \bibinfo{publisher}{Springer}, pp. \bibinfo{pages}{126--138},
  \doi{10.1007/978-3-540-77395-5_11}.

\bibitemdeclare{article}{bauercomparing}
\bibitem[BLS10]{bauercomparing}
\bibinfo{author}{Andreas \surnamestart Bauer\surnameend},
  \bibinfo{author}{Martin \surnamestart Leucker\surnameend} \&
  \bibinfo{author}{Christian \surnamestart Schallhart\surnameend}
  (\bibinfo{year}{2010}): \emph{\bibinfo{title}{Comparing LTL Semantics for
  Runtime Verification}}.
\newblock {\slshape \bibinfo{journal}{Journal of Logic and Computation}}
  \bibinfo{volume}{20}(\bibinfo{number}{3}), pp. \bibinfo{pages}{651--674},
  \doi{10.1093/logcom/exn075}.

\bibitemdeclare{article}{ltl3tosem}
\bibitem[BLS11]{ltl3tosem}
\bibinfo{author}{Andreas \surnamestart Bauer\surnameend},
  \bibinfo{author}{Martin \surnamestart Leucker\surnameend} \&
  \bibinfo{author}{Christian \surnamestart Schallhart\surnameend}
  (\bibinfo{year}{2011}): \emph{\bibinfo{title}{Runtime Verification for LTL
  and TLTL}}.
\newblock {\slshape \bibinfo{journal}{ACM Transactions on Software Engineering
  Methodology}} \bibinfo{volume}{20}(\bibinfo{number}{4}),
  \doi{10.1145/2000799.2000800}.

\bibitemdeclare{inproceedings}{safetyprogress}
\bibitem[CMP93]{safetyprogress}
\bibinfo{author}{Edward \surnamestart Chang\surnameend}, \bibinfo{author}{Zohar
  \surnamestart Manna\surnameend} \& \bibinfo{author}{Amir \surnamestart
  Pnueli\surnameend} (\bibinfo{year}{1993}): \emph{\bibinfo{title}{The
  Safety-Progress Classification}}.
\newblock In: {\slshape \bibinfo{booktitle}{Logic and Algebra of
  Specification}}, \bibinfo{publisher}{Springer}, pp.
  \bibinfo{pages}{143--202}, \doi{10.1007/978-3-642-58041-3_5}.

\bibitemdeclare{inproceedings}{chai}
\bibitem[CS14]{chai}
\bibinfo{author}{Ming \surnamestart Chai\surnameend} \&
  \bibinfo{author}{Bernd-Holger \surnamestart Schlingloff\surnameend}
  (\bibinfo{year}{2014}): \emph{\bibinfo{title}{Online Monitoring of
  Distributed Systems with a Five-Valued LTL}}.
\newblock In: {\slshape \bibinfo{booktitle}{IEEE 44th International Symposium
  on Multiple-Valued Logic}}, pp. \bibinfo{pages}{226--231},
  \doi{10.1109/ISMVL.2014.47}.

\bibitemdeclare{inproceedings}{ltlpm}
\bibitem[EFH{\etalchar{+}}03]{ltlpm}
\bibinfo{author}{Cindy \surnamestart Eisner\surnameend}, \bibinfo{author}{Dana
  \surnamestart Fisman\surnameend}, \bibinfo{author}{John \surnamestart
  Havlicek\surnameend}, \bibinfo{author}{Yoad \surnamestart Lustig\surnameend},
  \bibinfo{author}{Anthony \surnamestart McIsaac\surnameend} \&
  \bibinfo{author}{David \surnamestart Van~Campenhout\surnameend}
  (\bibinfo{year}{2003}): \emph{\bibinfo{title}{Reasoning with Temporal Logic
  on Truncated Paths}}.
\newblock In: {\slshape \bibinfo{booktitle}{Computer Aided Verification}},
  \bibinfo{publisher}{Springer}, pp. \bibinfo{pages}{27--39},
  \doi{10.1007/978-3-540-45069-6_3}.

\bibitemdeclare{inproceedings}{quant}
\bibitem[HMS23]{quant}
\bibinfo{author}{Thomas~A. \surnamestart Henzinger\surnameend},
  \bibinfo{author}{Nicolas \surnamestart Mazzocchi\surnameend} \&
  \bibinfo{author}{N.~Ege \surnamestart Sara{\c{c}}\surnameend}
  (\bibinfo{year}{2023}): \emph{\bibinfo{title}{Quantitative Safety and
  Liveness}}.
\newblock In: {\slshape \bibinfo{booktitle}{Foundations of Software Science and
  Computation Structures}}, \bibinfo{publisher}{Springer}, pp.
  \bibinfo{pages}{349--370}, \doi{10.1007/978-3-031-30829-1_17}.

\bibitemdeclare{inproceedings}{progress2}
\bibitem[KT05]{progress2}
\bibinfo{author}{Froduald \surnamestart Kabanza\surnameend} \&
  \bibinfo{author}{Sylvie \surnamestart Thi{\'{e}}baux\surnameend}
  (\bibinfo{year}{2005}): \emph{\bibinfo{title}{Search Control in Planning for
  Temporally Extended Goals}}.
\newblock In: {\slshape \bibinfo{booktitle}{International Conference on
  Automated Planning and Scheduling}}, \bibinfo{publisher}{{AAAI}}, pp.
  \bibinfo{pages}{130--139}.

\bibitemdeclare{article}{kupfermanvardi}
\bibitem[KV01]{kupfermanvardi}
\bibinfo{author}{Orna \surnamestart Kupferman\surnameend} \&
  \bibinfo{author}{Moshe~Y. \surnamestart Vardi\surnameend}
  (\bibinfo{year}{2001}): \emph{\bibinfo{title}{Model Checking of Safety
  Properties}}.
\newblock {\slshape \bibinfo{journal}{Formal Methods in System Design}}
  \bibinfo{volume}{19}(\bibinfo{number}{3}), pp. \bibinfo{pages}{291--314},
  \doi{10.1023/A:1011254632723}.

\bibitemdeclare{article}{Lam77}
\bibitem[Lam77]{Lam77}
\bibinfo{author}{Leslie \surnamestart Lamport\surnameend}
  (\bibinfo{year}{1977}): \emph{\bibinfo{title}{Proving the correctness of
  multiprocess programs}}.
\newblock {\slshape \bibinfo{journal}{IEEE Transactions on Software
  Engineering}} \bibinfo{volume}{3}(\bibinfo{number}{2}), pp.
  \bibinfo{pages}{125--143}, \doi{10.1109/TSE.1977.229904}.

\bibitemdeclare{inproceedings}{fltl}
\bibitem[LPZ85]{fltl}
\bibinfo{author}{Orna \surnamestart Lichtenstein\surnameend},
  \bibinfo{author}{Amir \surnamestart Pnueli\surnameend} \&
  \bibinfo{author}{Lenore \surnamestart Zuck\surnameend}
  (\bibinfo{year}{1985}): \emph{\bibinfo{title}{The Glory of the Past}}.
\newblock In: {\slshape \bibinfo{booktitle}{Logics of Programs}},
  \bibinfo{publisher}{Springer}, pp. \bibinfo{pages}{196--218},
  \doi{10.1007/3-540-15648-8_16}.

\bibitemdeclare{book}{ltl}
\bibitem[MP92]{ltl}
\bibinfo{author}{Zohar \surnamestart Manna\surnameend} \& \bibinfo{author}{Amir
  \surnamestart Pnueli\surnameend} (\bibinfo{year}{1992}):
  \emph{\bibinfo{title}{The Temporal Logic of Reactive and Concurrent
  Systems}}.
\newblock \bibinfo{publisher}{Springer}, \doi{10.1007/978-1-4612-0931-7}.

\bibitemdeclare{book}{ltlsafety}
\bibitem[MP95]{ltlsafety}
\bibinfo{author}{Zohar \surnamestart Manna\surnameend} \& \bibinfo{author}{Amir
  \surnamestart Pnueli\surnameend} (\bibinfo{year}{1995}):
  \emph{\bibinfo{title}{Temporal Verification of Reactive Systems: Safety}}.
\newblock \bibinfo{publisher}{Springer}, \doi{10.1007/978-1-4612-4222-2}.

\bibitemdeclare{inproceedings}{quickstrom}
\bibitem[OW22]{quickstrom}
\bibinfo{author}{Liam \surnamestart O'Connor\surnameend} \&
  \bibinfo{author}{Oskar \surnamestart Wickstr\"{o}m\surnameend}
  (\bibinfo{year}{2022}): \emph{\bibinfo{title}{Quickstrom: Property-based
  Acceptance Testing with LTL Specifications}}.
\newblock In: {\slshape \bibinfo{booktitle}{Programming Language Design and
  Implementation}}, \bibinfo{series}{PLDI 2022}, \bibinfo{publisher}{ACM}, p.
  \bibinfo{pages}{1025–1038}, \doi{10.1145/3519939.3523728}.

\bibitemdeclare{article}{rosu}
\bibitem[RH05]{rosu}
\bibinfo{author}{Grigore \surnamestart Ro\c{s}u\surnameend} \&
  \bibinfo{author}{Klaus \surnamestart Havelund\surnameend}
  (\bibinfo{year}{2005}): \emph{\bibinfo{title}{Rewriting-Based Techniques for
  Runtime Verification}}.
\newblock {\slshape \bibinfo{journal}{Automated Software Engineering}}
  \bibinfo{volume}{12}(\bibinfo{number}{2}), pp. \bibinfo{pages}{151--197},
  \doi{10.1007/s10515-005-6205-y}.

\bibitemdeclare{inproceedings}{rltl}
\bibitem[TN16]{rltl}
\bibinfo{author}{Paulo \surnamestart Tabuada\surnameend} \&
  \bibinfo{author}{Daniel \surnamestart Neider\surnameend}
  (\bibinfo{year}{2016}): \emph{\bibinfo{title}{{Robust Linear Temporal
  Logic}}}.
\newblock In: {\slshape \bibinfo{booktitle}{25th EACSL Annual Conference on
  Computer Science Logic, \rm CSL 2016}}, {\slshape \bibinfo{series}{Leibniz
  International Proceedings in Informatics (LIPIcs)}}~\bibinfo{volume}{62},
  \bibinfo{publisher}{Schloss Dagstuhl -- Leibniz-Zentrum f{\"u}r Informatik},
  pp. \bibinfo{pages}{10:1--10:21}, \doi{10.4230/LIPIcs.CSL.2016.10}.

\end{thebibliography}

\end{document}